\documentclass[conference]{IEEEtran}
\IEEEoverridecommandlockouts
\usepackage{cite}
\usepackage{amsmath,amssymb,amsfonts,amsthm}
\usepackage{bm}
\usepackage{algorithmic}
\usepackage{graphicx}
\usepackage{textcomp}
\usepackage{xcolor}
\usepackage{cuted}
\usepackage[font=small,tableposition=top]{caption}
\usepackage[braket]{qcircuit}
\usepackage{lscape}
\usepackage{pgfplotstable}
\usepackage{filecontents}
\usepackage[T1]{fontenc}
\usepackage{booktabs}
\usepackage{longtable}
\usepackage{tabu}
\usepackage{url}
\usepackage{breakurl}
\usepackage[breaklinks]{hyperref}
\usepackage{orcidlink}
\usepackage{thmtools}
\usepackage{fancyhdr}
\usepackage{tcolorbox}
\usepackage{verbatim}
\usepackage{dblfloatfix}
\usepackage{enumitem}

\def\BibTeX{{\rm B\kern-.05em{\sc i\kern-.025em b}\kern-.08em
    T\kern-.1667em\lower.7ex\hbox{E}\kern-.125emX}}

\hypersetup{
  breaklinks=true,
  colorlinks=true,
  linkcolor=blue,
  urlcolor=blue,
  citecolor=blue
}

\makeatletter
\newcommand{\linebreakand}{%
  \end{@IEEEauthorhalign}
  \hfill\mbox{}\par
  \mbox{}\hfill\begin{@IEEEauthorhalign}
}
\makeatother

\newcommand{\beq}[0]{\begin{equation}}
\newcommand{\eeq}[0]{\end{equation}}

\tcbuselibrary{theorems}
\newtcbtheorem[]{mydef}{Definition}%
{colback=blue!5,colframe=blue!35!white,fonttitle=\bfseries}{def}
\newtcbtheorem[number within=section]{mytheo}{Theorem}%
{colback=red!5,colframe=red!35!white,fonttitle=\bfseries}{th}
\newtcbtheorem[number within=section]{mylem}{Lemma}%
{colback=orange!5,colframe=orange!50!white,fonttitle=\bfseries}{th}
\newtcbtheorem[number within=section]{myprop}{Proposition}%
{colback=green!5,colframe=green!80!white,fonttitle=\bfseries}{th}

\newtheorem{theorem}{Theorem}[section]

\newtheorem{corollary}{Corollary}[theorem]

\theoremstyle{definition}

\begin{document}

\title{Qubit-efficient and gate-efficient encodings of graph partitioning problems for quantum optimization
\thanks{TZ, PPA, US, and HAM were funded by a National Sciences and Engineering Research Council of Canada (NSERC) Collaborative Research and Training Experience (CREATE) grant on Quantum Computing, NSERC Alliance Consortium Grant entitled Quantum Software Consortium -- Exploring Distributed Quantum Solutions for Canada (QSC), and NSERC Alliance grant on Quantum Computing for Optimal Mobility. TZ was additionally funded by the Simons Foundation.  VKM was wholly funded by the Simons Foundation. The computations reported in this paper were performed using an allocation on a D-Wave quantum annealer funded by the Simons Foundation. The Flatiron Institute is a division of the Simons Foundation.}
}

\author{
\IEEEauthorblockN{Tristan Zaborniak \orcidlink{0000-0002-4301-0861}}
\IEEEauthorblockA{\textit{Department of Computer Science} \\
\textit{University of Victoria}\\
Victoria, BC, Canada \\
\href{mailto:tristanz@uvic.ca}{tristanz@uvic.ca}}
\and
\IEEEauthorblockN{Prashanti Priya Angara \orcidlink{0000-0002-9660-011X}}
\IEEEauthorblockA{\textit{Department of Computer Science} \\
\textit{University of Victoria}\\
Victoria, BC, Canada \\
\href{mailto:pangara@uvic.ca}{pangara@uvic.ca}}
\and
\IEEEauthorblockN{Vikram Khipple Mulligan \orcidlink{0000-0001-6038-8922}}
\IEEEauthorblockA{\textit{Center for Computational Biology} \\
\textit{Flatiron Institute}\\
New York, NY, United States of America \\
\href{mailto:vmulligan@flatironinstitute.org}{vmulligan@flatironinstitute.org}}
\linebreakand
\IEEEauthorblockN{Hausi A. M\"uller \orcidlink{0000-0002-1004-5830}}
\IEEEauthorblockA{\textit{Department of Computer Science} \\
\textit{University of Victoria}\\
Victoria, BC, Canada \\
\href{mailto:hausi@uvic.ca}{hausi@uvic.ca}}
\and
\IEEEauthorblockN{Ulrike Stege \orcidlink{0000-0001-9466-7196}}
\IEEEauthorblockA{\textit{Department of Computer Science} \\
\textit{University of Victoria}\\
Victoria, BC, Canada \\
\href{mailto:ustege@uvic.ca}{ustege@uvic.ca}}
}

\maketitle

 \begin{abstract}
We introduce a qubit- and gate-efficient higher-order unconstrained binary optimization (HUBO) encoding for graph partitioning problems requiring label-count minimization. This widely applicable class of problems includes minimum graph coloring, minimum \textit{k}-cut, and community detection. To the best of our knowledge, this is the first work to address the optimization versions of these problems in a quantum setting, rather than only their decision counterparts. Our construction encodes each vertex variable in a number of bits logarithmic in the number of available labels---the information-theoretic minimum---and employs a novel lexicographic penalty system that implicitly minimizes partition count without requiring dedicated indicator variables. We derive provably sufficient conditions on all penalty coefficients, including those arising from Rosenberg quadratization, guaranteeing feasibility and optimality of the lowest-energy solution. Analogous conditions are derived for a one-hot encoding to enable controlled comparison. We also show that our encoding reduces two-qubit gate count per layer of the quantum approximate optimization algorithm (QAOA). Benchmarking on a quantum annealer demonstrates that our logarithmic encoding greatly improves solution quality and time-to-solution for minimum graph coloring relative to one-hot encoding, with greater advantage as problem size increases.
 \end{abstract}


\section{Introduction}\label{section:introduction}

Solving $\mathcal{NP}$-hard combinatorial optimization problems (COPs) using quantum computers has attracted considerable interest, as approximation advantages versus classical techniques are achievable in principle \cite{superpolynomial, Hartung2025} and have been observed in practice \cite{MunozBauza2025, Shaydulin2024}. Realizing these advantages, however, critically depends on how efficiently such problems are encoded for execution on quantum hardware. In particular, it is crucial to minimize qubit requirements and gate complexity across hardware generations, as these factors directly impact error rates, quantum error correction overhead, and, ultimately, the feasibility of solving practical problems at scale \cite{preskill2018, fowler2012surface}.

Here, we focus on reducing the qubit and gate requirements for a subset of COPs whose variables take symmetric labels in a finite set of size $k>2$. Such higher-arity variables arise naturally in labeling, clustering, and partitioning tasks, and their structure poses distinctive encoding challenges for quantum optimization, particularly when the cardinality of the assigned label set is itself to be minimized. A canonical example is \textsc{Minimum Graph Coloring}, in which each vertex is assigned exactly one color, with the objective of minimizing the number of distinct colors used while ensuring that adjacent vertices receive different colors \cite{jensen2011graph}. Many problems---including community detection \cite{Fortunato} and graph cutting \cite{Calinescu2008}---share the same underlying form: each vertex $v$ corresponds to a discrete variable $\ell_v \in \{0,\dots,k-1\}$ for some upper bound $k$, and the cost function depends on equality or inequality relations between the labels of adjacent vertices.

Efficiently encoding such $k$-valued variables in binary form to allow their compatibility with quantum optimization algorithms is nontrivial. Quantum optimization schemes typically operate on binary degrees of freedom \cite{kadowaki1998quantum, farhi2014quantum, Motta2019, Gilliam2021, nn2025}, necessitating a mapping from $k$-ary choices to bitstrings. The most common such mapping is the one-hot encoding \cite{Lucas2014}, which represents each label with $k$ indicator bits subject to an exclusivity constraint. While amenable to quadratic penalty formulations, one-hot encoding incurs a linear overhead in $k$ for every vertex. This quickly becomes a bottleneck, particularly on near-term hardware: a graph with $|V|$ vertices and $k$ allowed labels requires at least $|V|\cdot k$ binary variables (or qubits). Moreover, one-hot encoding introduces infeasible states and the need for carefully tuned penalty parameters \cite{Verma2022}.

As a result, the search for qubit-efficient and gate-efficient encodings of $k$-valued decision variables is ongoing.\footnote{By \textit{qubit-efficient}, we mean that each $k$-valued decision variable is encoded using $\lceil \log_2 k \rceil$ qubits, the information-theoretic minimum. By \textit{gate-efficient}, we mean, informally, circuit implementations requiring asymptotically fewer two-qubit gates per variational layer or oracle call than one-hot encodings.} Alternatives to one-hot encodings include domain-wall encodings \cite{chancellor2019domain}, compact encodings using $\lceil \log_2 k\rceil$ qubits per $k$-ary variable \cite{tabi_graph_coloring, glos_space_efficient, Zaborniak2025}, compact encodings using $\lceil\log_2 m\rceil$ qubits for binary optimization functions expressed natively in terms of $m$ bits \cite{Tan2021, Huber2024, Leonidas2024, Perelshtein2023}, compact encodings using between $\lceil\log_2 m\rceil$ and $m$ qubits for binary optimization functions expressed natively in terms of $m$ bits \cite{Sciorilli2025}, compressed Hilbert-space representations \cite{shirai_constrained_compression}, and other problem- and algorithm-aware constructions that trade qubit count for interaction locality, order, and/or query complexity \cite{sano2024accelerating, Chatterjee2024}. 

Each of these alternative encodings reduce qubit overhead, but suffer from one or more of the following problems: (i) introduction of many higher-order interactions, (ii) more complex feasibility constraints, (iii) considerably increased circuit depth, (iv) considerably increased sampling costs, (v) introduction of error to the objective function, and (vi) non-trivial adaptation or non-applicability to graph partitioning problems with minimum partition count objectives. In fact, only Ref.~\cite{tabi_graph_coloring} addresses a graph partitioning problem directly, treating the \textit{decision} version of graph coloring rather than the \textit{optimization} version that seeks to minimize the number of colors used. Taken together, these limitations highlight that the interplay among encoding size, feasible-state geometry, and penalty calibration is central to designing scalable optimization models for $k$-valued graph problems, including graph partitioning. We examine these trade-offs in Section \ref{section:background}.

Here, we introduce a qubit- and gate-efficient HUBO formulation for label-symmetric graph partitioning problems whose objective includes an explicit preference for fewer labels. To the best of our knowledge, our work is the first to treat such problems. Our construction assigns $\lceil\log_2 k\rceil$ bits per vertex and employs a novel lexicographic penalty system that imposes a strict hierarchy among label indices. This enforces a structured ordering that guarantees that feasible solutions favor smaller labels, implicitly capturing the minimum partition count objective without requiring separate indicator variables. The resulting HUBO admits quadratization with explicit bounds on auxiliary variables, making the approach compatible with annealing-based and gate-based quantum optimizers. We provide a comprehensive analysis of our formulation, including:

\begin{enumerate}[topsep=2pt,itemsep=1pt,parsep=0pt,partopsep=0pt]
    \item[\textnormal{(i)}] Explicit Hamiltonians for general label-symmetric, pairwise-decomposable graph partitioning problems with minimum partition count objectives;
    \item[\textnormal{(ii)}] Theoretical scaling analysis of qubit count, gate count, and penalty weights, including under Rosenberg order reduction; and
    \item[\textnormal{(iii)}] Benchmarking results demonstrating the superiority of our approach over one-hot encoding for \textsc{Minimum Graph Coloring} on a quantum annealer.
\end{enumerate}

\noindent The resulting framework applies broadly across this class of graph partitioning problems and across existing quantum optimization schemes and hardware.

\section{Background}\label{section:background}

We now discuss: (i) current approaches to qubit count reduction when encoding COPs, (ii) current quantum combinatorial optimization algorithms compatible with our encodings, and (iii) the class of problems that we treat explicitly in this paper.

\subsection{Qubit-Efficient Encodings of COPs}

The standard approach to encoding COPs on quantum hardware employs one-hot encoding \cite{Lucas2014}. This encoding requires $k$ qubits for each of $|V|$ discrete, $k$-ary variables, leading to quantum circuits with at least $|V|\cdot k$ qubits. As noted above, there has been effort to develop more qubit-efficient encodings that reduce this overhead, often at the cost of increased circuit depth, higher-order interactions, or more complex feasibility constraints. Here we briefly review these approaches, highlighting their trade-offs and practical limitations.

\paragraph{Domain-Wall Encoding} Domain-wall encoding represents a $k$-valued variable using $k-1$ qubits arranged such that the bit string $1^{\ell}0^{k-1-\ell}$ corresponds to value $\ell$ \cite{chancellor2019domain}. This encoding offers a modest qubit reduction over one-hot encoding (from $k$ to $k-1$ qubits per variable), and a modest reduction in the number of infeasible states. Importantly, as for one-hot encodings, constraints such as equality or inequality between two domain-wall variables can be expressed using only quadratic terms \cite{chancellor2019domain, Berwald2023}. The primary limitation of domain-wall encoding is therefore that the qubit reduction is minimal. Consequently, it does not provide the exponential or polynomial compression necessary to tackle large-scale instances on hardware-limited quantum processors \cite{chen2021performance}.

\paragraph{Logarithmic Encodings} A more aggressive compression strategy encodes each $k$-valued variable into $\lceil \log_2 k \rceil$ qubits using a binary representation \cite{tabi_graph_coloring, sano2024accelerating, glos_space_efficient}. This exponential reduction in qubit count is attractive, but produces higher-order objective functions. For example, when \textsc{Graph Coloring} is encoded in this manner, the resulting HUBO objective function has a maximum degree of $2\lceil \log_2 k \rceil$ \cite{glos_space_efficient}. Quadratizing these higher-order terms \textit{via} standard techniques introduces $\mathcal{O}(\lceil \log_2 k \rceil)$ auxiliary variables per constraint.

A related approach designed to avoid higher-order overhead compresses $k$-ary variables into $\lceil \log_2 k \rceil$ qubits by fitting the original cost function to QUBO form \textit{via} overdetermined least-squares \cite{Zaborniak2025}. While this achieves exponential qubit reduction, the encoding is lossy: the resulting objective function is an approximation of the original, with approximation error that grows with the degree of compression.

We also mention a critical limitation of logarithmic encodings when applied to $m$ \textit{binary} decision variables (applicable to $k$-ary problems only after encoding to binary variables) \cite{Tan2021, Huber2024, Leonidas2024, Perelshtein2023}. Specifically, such encodings are classically tractable when variational circuit depth remains polynomial in the number of qubits \cite{Tan2021, Sciorilli2025}. Moreover, such minimal encodings strongly limit the expressivity of the model\cite{Sciorilli2025}.

\paragraph{Polynomial Encodings} A recent approach by Sciorilli \textit{et al.} \cite{Sciorilli2025} encodes $m = \mathcal{O}(n^k)$ binary decision variables into correlations of $k$-body Pauli operators across $n$ qubits, achieving polynomial compression. This compression is achieved by defining each variable $x_i = \mathrm{sign}(\langle \Pi_i \rangle)$, where $\Pi_i$ is a traceless Pauli string. The quantum circuit is trained to minimize a non-linear loss function of the Pauli expectation values. This encoding offers several advantages: the circuit depth scales sublinearly in $n$, the number of parameters grows approximately linearly with $m$, and the scheme super-polynomial mitigates barren plateaus \cite{Sciorilli2025}.

However, while the encoding requires only three measurement settings ($X^{\otimes k}$, $Y^{\otimes k}$, $Z^{\otimes k}$), the measurement budget can be substantial (as is true of the approach in Ref.~\cite{Chatterjee2024}): worst-case bounds indicate $\mathcal{O}(m^3)$ measurements are required to estimate the loss function at each stage. Second, the method requires careful hyperparameter tuning. Finally, while the scheme is claimed to be classically intractable for sufficiently large $|V|$, the crossover point where quantum advantage emerges remains empirically unclear.

\subsection{Quantum Optimization Algorithms}

Quantum algorithms for combinatorial optimization that are compatible with our encoding include variational quantum computing (VQC) \cite{cerezo2021variational}, quantum imaginary time evolution (QITE) \cite{Motta2019}, Grover-based adaptive search (GAS) \cite{Gilliam2021, Ominato2024}, and adiabatic quantum computation (AQC)/quantum annealing (QA) \cite{mcgeoch2022adiabatic}. We summarize these methods, and briefly comment on their resource use when applied QUBO and HUBO models.

\paragraph{Variational Quantum Computing} VQC minimizes $\langle\psi(\boldsymbol{\theta})|H|\psi(\boldsymbol{\theta})\rangle\geq E_0$ by classically adjusting circuit parameters $\boldsymbol{\theta}$ \cite{cerezo2021variational}. QUBOs require only $2$-qubit gates; $k$-local HUBO terms are realized \textit{via} phase gadgets using $2(k{-}1)$ CNOTs and one $R_Z$ rotation per interaction \cite{cowtan2019phase, farhi2014quantum}.

\paragraph{Quantum Imaginary Time Evolution} QITE projects toward the ground state by approximating non-unitary imaginary-time dynamics with local unitaries \cite{Motta2019}. Resource scaling depends on estimating $O(n^k)$ correlators in the worst case for a $k$-local Hamiltonian, though locality and commutation relations can reduce this cost \cite{mcardle2019variational}. 

\paragraph{Grover Adaptive Search} GAS iteratively applies Grover amplification with thresholds set by the current best objective value \cite{Grover1996, Gilliam2021}. The oracle circuit depth scales as $O(n^2)$ for QUBOs and $O(n^k)$ for degree-$k$ HUBOs. Effectiveness is constrained by oracle circuit depth, which increases rapidly with cost function degree \cite{Gilliam2021}.

\paragraph{Adiabatic Quantum Computing / Quantum Annealing} AQC interpolates from an initial Hamiltonian $H_0$ to a problem Hamiltonian $H_1$ whose ground state encodes the optimum \cite{mcgeoch2022adiabatic, kadowaki1998quantum}. Practical devices typically support only quadratic Ising models, requiring higher-order problems to be reduced to quadratic order before execution. Minimizing qubit count is especially important here, as the inverse spectral gap---which is proportional to Time-to-Solution---typically grows exponentially with qubit number \cite{Morita2008}.

\subsection{Graph Partitioning Problems}\label{sec:graph_paritioning_problems}

The specifc COPs we treat in this work have the following properties. Given an undirected, simple graph $G=(V,E)$:

\begin{enumerate}[topsep=2pt,itemsep=1pt,parsep=0pt,partopsep=0pt]
    \item[\textnormal{(i)}] \textbf{Vertex labelling.} Each vertex $v\in V$ must be assigned a label $\ell_v\in\{0,1,\dots,k-1\}$ (\textit{i.e.}, $[k]$), where $k\leq|V|$.
    \item[\textnormal{(ii)}] \textbf{Pairwise decomposability.} There exists a pairwise cost or constraint function $\phi_{u,v}(\ell_u,\ell_v)$ for each edge $(u,v)\in E$ such that $\phi_{u,v}(\ell_u,\ell_v)=\phi_{u,v}(\mathbf{1}[\ell_u=\ell_v])$, where $\mathbf{1}[\cdot]$ is the indicator function.
    \item[\textnormal{(iii)}] \textbf{Label symmetry.} The labels are exchangeable; \textit{i.e.}, there is no inherent meaning to label $0$ versus label $1$, besides the fact that they are different from one another. 
    \item[\textnormal{(iv)}] \textbf{Minimum partition count.} The objective is to minimize the number of distinct labels assigned to vertices, assigning only one label per vertex.
\end{enumerate}

\noindent We therefore call these problems \textit{label-symmetric, pairwise-decomposable graph partitioning problems with minimization of partition count}. The general form of these problems seeks to minimize the following objective:

\begin{equation}\label{equation:general_graph_partitioning_problem}
    \min_{\ell:V\rightarrow [k]}\sum_{(u,v) \in E} \phi_{u,v}(\mathbf{1}[\ell_u= \ell_v]) + \lambda\cdot|\{\ell_v : v \in V\}|
\end{equation}

\noindent where $|\{\ell_v:v\in V\}|$ represents the number of distinct labels assigned to the full set of vertices, and the form of $\phi_{u,v}$ is such that we can write it as follows:

\begin{equation}
    \phi_{u,v}(\mathbf{1}[\ell_u= \ell_v])=\alpha_{u,v}\cdot\mathbf{1}[\ell_u= \ell_v]+\beta_{u,v}\cdot(\mathbf{1}[\ell_u\neq \ell_v])
\end{equation}

\noindent where $\alpha_{u,v}$ and $\beta_{u,v}$ encode \textit{local} properties between vertices that do not depend on label assignments. For example, in the case of \textsc{Minimum Graph Coloring}, $\beta_{u,v}=0\;\forall\;(u,v)\in E$, and $\alpha_{u,v}=1$ when $(u,v)\in E$ and $\alpha_{u,v}=0$ otherwise.

\section{Formulating \textsc{Minimum Graph Coloring} using a Logarithmic Number of Bits}\label{section:min_graph_coloring}

In this section, we develop our formulation for \textsc{Minimum Graph Coloring}, deferring the general case to Section \ref{section:treating_the_general_case}. We present both a one-hot QUBO and our logarithmic HUBO formulation, deriving sufficient conditions on all penalty coefficients---including those from Rosenberg quadratization---to guarantee feasibility of the lowest-energy solution. We also compare the two-qubit gate counts of the resulting Hamiltonians.

\begin{mydef}[]{\textsc{Minimum Graph Coloring}}{minimum_graph_coloring}
\vspace{-5mm}
    \begin{align*}
        &\text{Input:} \qquad G=(V,E) \text{~with~} |V| = n \text{~and~} |E| = m  \\
        &\text{Minimize:} \ |C| \\
        &\text{Where:} \quad\;
        \begin{cases}
            & C \text{~is a family of subsets of~} V, \text{~with:}\\
            &(1)\;\bigcup\limits_{V_i\in C} V_i = V\\
            &(2)\;\;\forall\; V_i,V_j \in C, i\not= j, V_i\cap V_j = \emptyset \\
            &(3)\;\;\forall\; (u,v)\in E, \text{if }\exists\, V_i\in C\\&\quad\;\; \text{~with~} u\in V_i,\mbox{~then~} v\notin V_i
        \end{cases}
    \end{align*}
    We call any family $C$ of subsets of $V$, which satisfies conditions (1)--(3), a \text{coloring} of $G$.
\end{mydef}

We begin with an arbitrary undirected, simple graph $G=(V,E)$. The \textsc{Minimum Graph Coloring} problem is then defined as in Definition \ref{def:minimum_graph_coloring}. See Figure \ref{fig:encodings} for an example of a graph hosting a proper coloring and an improper coloring.

\begin{figure}
\begin{center}
    \includegraphics[width=8.5cm]{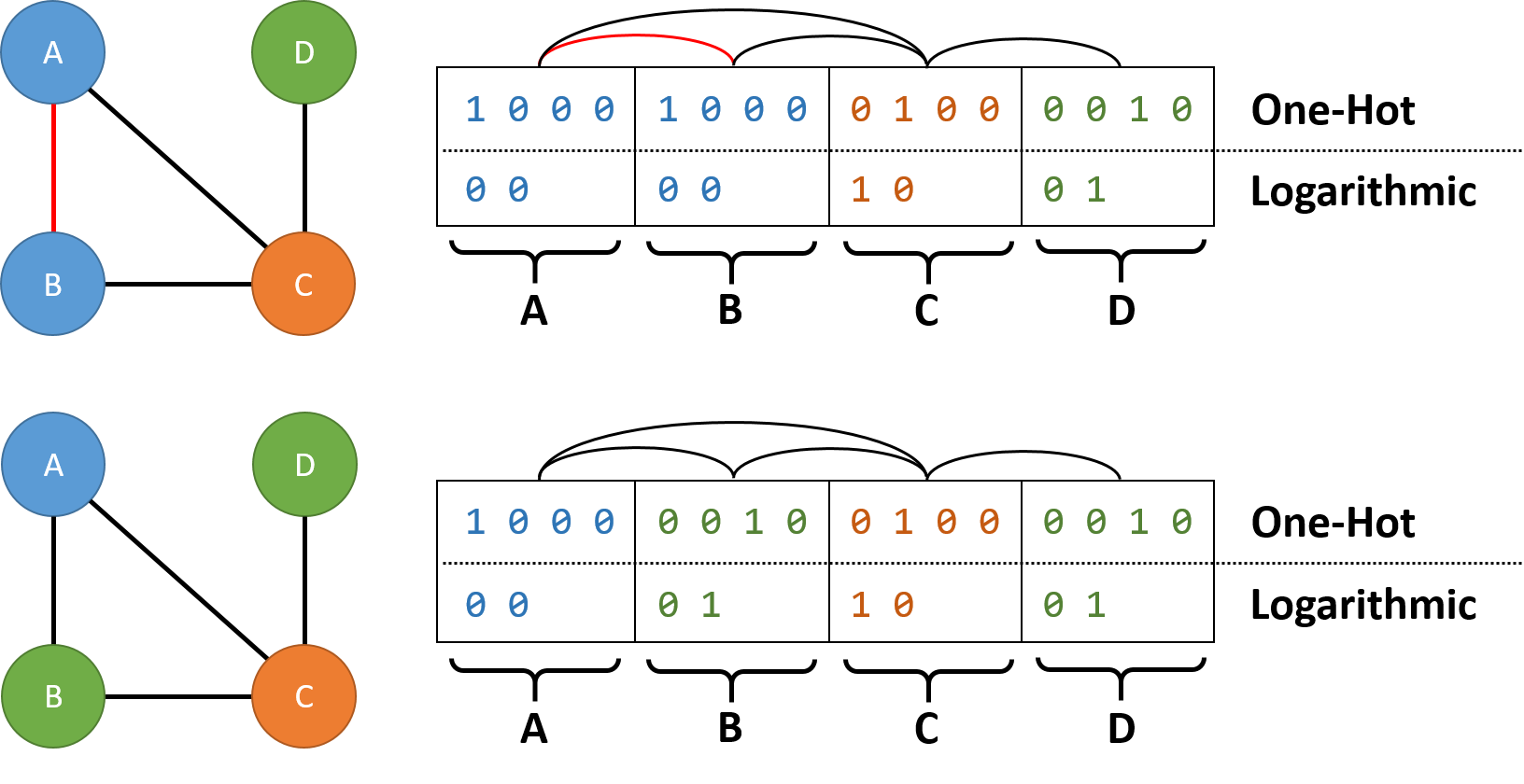}
    \caption{Small example graph colored improperly (upper row, violating edge colored red) and properly (bottom row). The tables contain bitstring assignments to each vertex under the one-hot and logarithmic encodings. Note that with $4$ ($|V|$) vertices, $4$ ($|V|$) qubits are required per vertex for one-hot encoding, and $2$ ($\lceil\log_2|V|\rceil$) qubits are required for logarithmic encoding. While the bitstrings are colored for visualization purposes, they do not refer to a \textit{particular} color.}
    \label{fig:encodings}
\end{center}
\end{figure}

\subsection{One-Hot \textsc{Minimum Graph Coloring} QUBO Encoding}\label{section:one_hot_formulation}

In Lucas, $2014$ \cite{Lucas2014}, a QUBO model for the decision problem, \textsc{Graph Coloring} (GC), is presented, which asks when given a graph $G$ and an integer $C_\textnormal{num}=|C|>0$ \textit{if} $G$ can have its vertices colored---using no more than $C_\textnormal{num}$ different colors---in such a way that no two vertices sharing an edge are colored identically. The QUBO model is as follows:

\begin{equation}\label{equation:min_graph_coloring_qubo}
H^\textnormal{one-hot}_\text{GC}(\mathbf{x}) = H_\textnormal{one-hot}(\mathbf{x}) + H_\textnormal{adjacency}(\mathbf{x})
\end{equation}

\noindent where:

\begin{align}\label{equation:min_graph_coloring_HV}
H_\textnormal{one-hot}(\mathbf{x}) 
&=A_\textnormal{one-hot}
\sum_v\left(1-\sum_{c=1}^{C_\textnormal{num}} x_{v, c}\right)^2 \\
 H_\textnormal{adjacency}(\mathbf{x}) 
&=A_\textnormal{adjacency}    \sum_{(u, v) \in E} \sum_{c=1}^{C_\textnormal{num}} x_{u, c} x_{v, c}
\end{align}

\noindent Here, $x_{v,c}$ is a binary variable which is $1$ if vertex $v$ is colored with color $c$, and $0$ otherwise. $A_\textnormal{one-hot}$ and $A_\textnormal{adjacency}$ are penalties that must be set to ensure that the lowest-energy solution is a coloring. $H_\textnormal{one-hot}$ uses the one-hot encoding $\sum_{c=1}^{C_\textnormal{num}} x_{v, c}$ and encourages that each vertex is colored only once. $H_\textnormal{adjacency}$ penalizes the case where two vertices share an edge that are colored identically. If we are testing $C_\textnormal{num}$ colors, this formulation requires $|V|\cdot C_\textnormal{num}$ binary variables.

In order to express \textsc{Minimum Graph Coloring} (MGC) as a QUBO, we set  $C_\textnormal{num}=C_\textnormal{num}^{u}$\footnote{Here, $C_\textnormal{num}^u \leq |V|$ denotes an efficiently-computable upper bound for $C_\textnormal{num}$, such as the one obtained by Brook's theorem \cite{Brooks1941}.} and add an additional register $\mathbf{y}$ of $C_\textnormal{num}$ binary variables  whose purpose is to indicate the number of colors used using one {\it indicator} variable per color. The function to minimize is then:

\begin{equation}\label{eqn:one-hot-MGC}
 H^\textnormal{one-hot}_\text{MGC}(\mathbf{x}, \mathbf{y}) =  H_\text{GC}(\mathbf{x}) + H_\text{count}(\mathbf{y}) + H_\text{link}(\mathbf{x},\mathbf{y})
\end{equation} 

\noindent where:

\begin{align}
    H_\text{count}(\mathbf{y}) & = \sum_{c=1}^{C_\textnormal{num}}y_c \label{equation:H-count} \\
    H_\text{link}(\mathbf{x},\mathbf{y}) & =A_\textnormal{link}\sum_v\sum_{c=1}^{C_\textnormal{num}}x_{v,c}(1-y_c) \label{equation:H-link}
\end{align}

\noindent $H_\text{count}$ is responsible for ensuring that the number of colors used is minimized under the given constraints.  $H_\text{link}$ enforces that if vertex $v$ is colored with color $c$, then $y_c$ is set to $1$. This QUBO formulation requires $(|V|+1)\cdot C_\textnormal{num}$ bits, and three independent penalties.

\subsection{Sufficient Conditions for One-Hot Formulation Penalties}\label{section:onehot_penalties}

In Equations \eqref{equation:min_graph_coloring_qubo}--\eqref{equation:H-link}, $H_\textnormal{one-hot}$ and $H_\textnormal{adjacency}$ encode the hard problem constraints, $H_\text{count}$ is the soft objective (corresponding to color count minimization), and $H_\text{link}$ links the color-usage indicators $y_c$ to the vertex assignments $x_{v,c}$.

We now present sufficient conditions on the penalty coefficients $A_\textnormal{one-hot}, A_\textnormal{adjacency}, A_\textnormal{link} > 0$ that guarantee the following properties for every lowest-energy solution $(\mathbf{x}^*, \mathbf{y}^*)$ of $H^\textnormal{one-hot}_\text{MGC}$:

\begin{enumerate}[topsep=2pt,itemsep=1pt,parsep=0pt,partopsep=0pt]
    \item[\textnormal{(i)}] \textbf{Indicator faithfulness.} Indicator variables faithfully track color usage: $y^*_c = 1$ if and only if $\sum_v x^*_{v,c} \geq 1$. 
    \item[\textnormal{(ii)}] \textbf{Proper coloring.} Edge-sharing vertices are colored differently: $x^*_{u,c}\,x^*_{v,c} = 0$ for all $(u,v) \in E$ and all $c$.
    \item[\textnormal{(iii)}] \textbf{One-hot constraint.} Each vertex is assigned exactly one color; \textit{i.e.}, $\sum_{c} x^*_{v,c} = 1$ for all $v$.
    \item[\textnormal{(iv)}] \textbf{Color-count minimality.} Among all solutions satisfying (i)--(iii), the number of distinct colors is minimized.
\end{enumerate}

\begin{theorem}[Sufficient One-Hot QUBO Penalties for \textsc{Minimum Graph Coloring}]\label{theorem:onehot_penalties}

Suppose the penalty coefficients satisfy the following system of inequalities:

\begin{align}
    A_\textnormal{link} &> 1 \label{equation:A4-bound} \\
    A_\textnormal{adjacency} &> A_\textnormal{link}\cdot C_\textnormal{num} \label{equation:A2-bound} \\
    A_\textnormal{one-hot} &> A_\textnormal{adjacency}\cdot|E| + A_\textnormal{link}\cdot C_\textnormal{num} \label{equation:A1-bound}
\end{align}

\noindent Then every lowest-energy solution $(\mathbf{x}^*, \mathbf{y}^*)$ of $H_\textnormal{MGC}(\mathbf{x}, \mathbf{y})$ satisfies properties \textnormal{(i)--(iv)}.
\end{theorem}

\begin{proof}
We prove (i)--(iv) in order.

\smallskip

\noindent\textbf{(i) Indicator faithfulness.}
Fix any $\mathbf{x}$ satisfying (iii) and any color $c_\alpha$. If $c_\alpha$ is used ($\sum_v x_{v,c_\alpha} \geq 1$), setting $y_{c_\alpha}=1$ costs $1$ \textit{via} $H_\text{count}$ but saves at least $A_\textnormal{link} > 1$ in $H_\text{link}$, so $y_{c_\alpha}=1$ is strictly preferred. If $c_\alpha$ is unused, setting $y_{c_\alpha}=0$ saves $1$ in $H_\text{count}$ with no change to $H_\text{link}$. Thus any optimum satisfying (iii) also satisfies (i).

\smallskip

\noindent\textbf{(ii) Proper coloring.}
Suppose $(\mathbf{x}^*,\mathbf{y}^*)$ satisfies (iii) but has a monochromatic edge: $x^*_{u,c_\alpha}=x^*_{v,c_\alpha}=1$ for some $(u,v)\in E$. Recolor $u$ to an unused color $c_\beta$, updating $\mathbf{y}$ to satisfy (i). The one-hot and link terms remain $0$, $H_\text{count}$ increases by at most $1$, and $H_\textnormal{adjacency}$ decreases by at least $A_\textnormal{adjacency}$. The net change is at most $1 - A_\textnormal{adjacency} < 0$ by Inequality \eqref{equation:A2-bound}, contradicting optimality.

\smallskip

\noindent\textbf{(iii) One-hot constraint.}
Suppose $(\mathbf{x}^*,\mathbf{y}^*)$ violates the one-hot constraint for some vertex, so $H_\textnormal{one-hot}(\mathbf{x}^*)\geq A_\textnormal{one-hot}$. Any one-hot-satisfying solution $(\tilde{\mathbf{x}},\tilde{\mathbf{y}})$ with faithful indicators has $H_\textnormal{one-hot}=0$, $H_\text{link}=0$, $H_\textnormal{adjacency} \leq m \cdot A_\textnormal{adjacency}$, and $H_\text{count}\leq C_\textnormal{num}$. Then $H^\textnormal{one-hot}_\text{MGC}(\tilde{\mathbf{x}},\tilde{\mathbf{y}}) \leq A_\textnormal{adjacency}\cdot |E| + C_\textnormal{num} < A_\textnormal{adjacency}\cdot |E| + C_\textnormal{num}\cdot A_\textnormal{link} < A_\textnormal{one-hot} \leq H^\textnormal{one-hot}_\text{MGC}(\mathbf{x}^*,\mathbf{y}^*)$, where the last inequality uses Inequality \eqref{equation:A1-bound}, contradicting optimality.

\smallskip

\noindent\textbf{(iv) Color-count minimality.}
By (i)--(iii), every optimum has $H_\textnormal{one-hot}=H_\textnormal{adjacency}=H_\text{link}=0$, so $H^\textnormal{one-hot}_\text{MGC}=H_\text{count}=|\{c:y_c^*=1\}|$, which is minimized exactly when the number of distinct colors is minimized.
\end{proof}

\begin{corollary}[Practical One-Hot QUBO Penalties for \textsc{Minimum Graph Coloring}]\label{cor:onehot_practical}

A convenient explicit choice satisfying the conditions of Theorem \ref{theorem:onehot_penalties} is:

\begin{align}
  % A_3 &= 1
  % \label{equation:oh-A3} \\
      A_\textnormal{link} &= |V| + 1
  \label{equation:oh-A4} \\
  A_\textnormal{adjacency} &= A_\textnormal{link}\cdot C_\textnormal{num} + 1 = (|V|+1)\cdot C_\textnormal{num} + 1
  \label{equation:oh-A2} \\
  A_\textnormal{one-hot} &= A_\textnormal{adjacency}\cdot (|E| + 1) + A_\textnormal{link}\cdot C_\textnormal{num}
  \label{equation:oh-A1}
\end{align}

\noindent since $A_\textnormal{link} = |V|+1 > 1$ satisfies Inequality \eqref{equation:A4-bound}; $A_\textnormal{adjacency} = A_\textnormal{link}\cdot C_\textnormal{num} + 1 > A_\textnormal{link}\cdot C_\textnormal{num}$ satisfies Inequality \eqref{equation:A2-bound}; and $A_\textnormal{one-hot} = A_\textnormal{adjacency}\cdot (|E|+1) + A_\textnormal{link}\cdot C_\textnormal{num} > A_\textnormal{adjacency}\cdot|E| + A_\textnormal{link}\cdot C_\textnormal{num}$ satisfies Inequality \eqref{equation:A1-bound}.
\end{corollary}

\subsection{Logarithmic \textsc{Minimum Graph Coloring} HUBO Encoding}\label{section:hubo_formulation}

Given that $C_\textnormal{num}$ colors can be enumerated using $\lceil\log_2 C_\textnormal{num}\rceil$ bits, we seek a cost function expressed using such a number of bits per vertex. In this way, we can avoid introducing infeasible solutions by way of one-hot encoding.

Focusing for now on the \textit{adjacency condition}---\textit{i.e.}, that no two adjacent vertices may be colored the same---we must penalize the case where bitstrings corresponding to adjacent vertices are identical. We accomplish this using a product of XNORs over edge bits, such that \textsc{Graph Coloring} then reduces to the following Hamiltonian \cite{tabi_graph_coloring}:

\begin{equation}\label{equation:k_coloring_hubo}
    H^\textnormal{logarithmic}_\textnormal{GC} (\mathbf{x})  = \sum_{(u, v) \in E}\prod_{k=1}^{\lceil\log_2 C_\textnormal{num}\rceil}\left(x_{u,k}\odot x_{v,k}\right)
\end{equation}

\noindent where $x_{v,k}$ refers to bit $k$ of vertex $v$, and $x_{u,k}\odot x_{v,k} =2x_{u,k}x_{v,k} - x_{u,k}- x_{v,k} +1$. When the bitstrings of $u$ and $v$ are identical, $\mathbf{x}_{u,k}\odot x_{v,k}=1$ for all $k$, such that the product equals 1. Therefore, valid colorings have scores of zero, and invalid colorings have scores greater than zero. This HUBO for \textsc{Graph Coloring} uses $|V|\cdot\lceil\log_2 C_\textnormal{num}\rceil$ bits, and has a maximum degree of $2\cdot\lceil\log_2{(C_\textnormal{num})}\rceil$.

When we would like to find the \textit{minimum} number of colors required to color some graph $G$, we can introduce $\lceil\log_2 C_\textnormal{num}\rceil$ linear penalty terms, increasingly strong for increasingly-significant bits within bitstrings. That is, if we assign \textit{significance} to the bits of each vertex bitstring, we can prioritize coloring with ``low-significance'' bitstrings, and therefore minimize the number of colors required. 

Designating the least significant bit by $k=1$ and the most significant bit by $k=\lceil\log_2 C_\textnormal{num}\rceil$, and ensuring that $P_1<P_2<\cdots<P_{\lceil\log_2 C_\textnormal{num}\rceil}$ (see Section \ref{section:sufficient_conditions_lex_and_feas} for strict bounds on $\{P_k\}_{k\in[\lceil\log_2 C_\textnormal{num}\rceil]}$), we  have the following \textsc{Minimum Graph Coloring} HUBO model:

\begin{align}\label{equation:min_graph_coloring_hubo}
    H^\textnormal{logarithmic}_\textnormal{MGC}(\mathbf{x})&=\underbrace{A_\textnormal{adjacency}\sum_{(u, v) \in E}\prod_{k=1}^{\lceil\log_2 C_\textnormal{num}\rceil}\left(x_{u,k}\odot x_{v,k}\right)}_{H_\text{adjacency}}\nonumber \\&\quad\;+ \underbrace{\sum_{k=1}^{\lceil\log_2 C_\textnormal{num}\rceil}P_k\sum_{v}x_{v,k}}_{H_\text{lexicographic}}
\end{align}

\noindent where $A_\textnormal{adjacency}$ is a tunable penalty. In the next section, we derive explicit bounds on $\{P_k\}_{k\in[1,...,\log_2 C_\textnormal{num}]}$ to ensure that \textit{lexicographic minimization} is preserved (Definition \ref{def:lexicographic_minimization}).

\begin{mydef}[]{\textsc{Lexicographic Minimization}}{lexicographic_minimization}

Given two ``index-population'' vectors:
\begin{equation*}
    \mathbf{s} = (s_L, s_{L-1}, \ldots, s_1),\quad \mathbf{s}' = (s'_L, s'_{L-1}, \ldots, s'_1)
\end{equation*}

\noindent where $s_k\in\mathbb{Z}^+\cup0$, we say $\mathbf{s} <_{\text{lexicographic}} \mathbf{s}'$ if at the leftmost index $k$ where they differ, we have $s_k < s'_k$. Formally:
\begin{align*}
    \mathbf{s} &<_{\text{lexicographic}} \mathbf{s}' \iff 
\nonumber \\ &\quad\exists k \in \{L, L-1, \ldots, 1\} : 
s_k < s'_k,\; s_j = s'_j\;\forall\; j > k
\end{align*}
\end{mydef}

To \textit{lexicographically minimize} $\mathbf{s}$ in the context of \textsc{Minimum Graph Coloring} means to find the label assignment (coloring) that is minimal under the $<_{\text{lexicographic}}$ ordering, where $s_k=\sum_{v\in V}x_{v,k}$.

\subsection{Sufficient Conditions for Logarithmic Formulation Penalties}\label{section:sufficient_conditions_lex_and_feas}

Fix an integer $L=\lceil\log_2 C_\textnormal{num}\rceil$  to be the \textit{length} of the bitstring associated with each vertex $v\in V$. For each $v$, and each index within a bitstring $k\in\{1,\dots,L\}$, $x_{v,k}\in\{0,1\}$. 

We now present sufficient conditions on the penalty coefficients $A_\textnormal{adjacency}, P_k\,\forall\,k> 0$ that guarantee the following properties for every lowest-energy solution $(\mathbf{x}^*, \mathbf{y}^*)$ of $H_\text{MGC}$:

\begin{enumerate}[topsep=2pt,itemsep=1pt,parsep=0pt,partopsep=0pt]
    \item[\textnormal{(i)}] \textbf{Lexicographic ordering.} Among all feasible solutions, $\mathbf{x}^*$ lexicographically minimizes the vector $(s_L,s_{L-1},\dots,s_1)$.
    \item[\textnormal{(ii)}] \textbf{Feasibility of the optimal solution.} $H_\textnormal{adjacency}(\mathbf{x}^*)=0$; \textit{i.e.}, no edge of $G$ connects two vertices with the same bitstring.
\end{enumerate}

\begin{theorem}[Sufficient Logarithmic HUBO Penalties for \textsc{Minimum Graph Coloring}]\label{theorem:penalties_coloring}

Suppose $P_1,\dots,P_L$ and $A_\textnormal{adjacency}$ satisfy the following two inequalities:

\begin{align}
P_{k+1} &> |V|\cdot\sum_{j=1}^{k} P_j 
,\;\forall\; k=1,\dots,L-1 \label{equation:penalty_sum}\\
A &> |V|\cdot\sum_{k=1}^{L} P_k\label{equation:adjacency}
\end{align}

\noindent Then any lowest-energy solution $\mathbf{x}^*$ of $H^\textnormal{one-hot}_\textnormal{MGC}(\mathbf{x})$ satisfies properties \textnormal{(i)--(ii)}.

\end{theorem}

\begin{proof}
We prove (i)--(ii) in order.

\smallskip

\noindent\textbf{(i) Lexicographic ordering.} Fix any two solutions $\mathbf{x}$ and $\tilde{\mathbf{x}}$ and write their index-population vectors as $(s_L,\dots,s_1)$ and $(s'_L,\dots,s'_1)$, respectively. Let $t$ be the largest index at which these vectors differ:

\begin{equation}
t = \max\{k: s_k\neq s'_k\}
\end{equation}

\noindent Without loss of generality assume $s_t' < s_t$. Then the change in the lexicographic penalty from $\mathbf{x}$ to $\tilde{\mathbf{x}}$ is:

\begin{equation}
\Delta H_\textnormal{lexicographic}= P_t\cdot(s'_t-s_t) + \sum_{j=1}^{t-1} P_j\cdot(s'_j-s_j)
\end{equation}

\noindent By construction, $s'_t-s_t \le -1$, so the first term is at most $-P_t$, and at least $-P_t\cdot(n-1)$. The remaining terms can at worst increase the energy. Then, each $|s'_j-s_j|$ is bounded above by $|V|$, such that:

\begin{equation}
- |V|\cdot\sum_{j=1}^{t-1} P_j\;\le\; \sum_{j=1}^{t-1} P_j\cdot(s'_j-s_j)\;\le \;|V|\cdot\sum_{j=1}^{t-1} P_j
\end{equation}

\noindent Combining these bounds yields:

\begin{align}
\Delta H_\textnormal{lexicographic}&\ge-P_t\cdot(n-1) - |V|\cdot\sum_{j=1}^{t-1} P_j\\ \Delta H_\textnormal{lexicographic} &\le -P_t + |V|\cdot\sum_{j=1}^{t-1} P_j
\end{align}

By Inequality \eqref{equation:penalty_sum} with $i=t-1$ we have $P_t > |V|\cdot\sum_{j=1}^{t-1}P_j$; hence, the right-hand side is strictly negative. Therefore $H_\textnormal{lexicographic}(\tilde{\mathbf{x}})<H_\textnormal{lexicographic}(\mathbf{x})$. Since $t$ was the most significant index at which the two vectors differed, this shows that any decrease in a more significant count $s_t$ always yields a net decrease in lexicographic energy that cannot be compensated by arbitrary increases in less significant counts. Consequently minimizing $H_\textnormal{lexicographic}$ is equivalent to lexicographically minimizing $(s_L,\dots,s_1)$, proving (i).

\smallskip

\noindent\textbf{(ii) Feasibility of the optimal solution.} Suppose for contradiction that a lowest-energy solution $\mathbf{x}^*$ has at least one violating edge; \textit{i.e.}, $H_\textnormal{adjacency}(\mathbf{x}^*)\ge1$. Construct $\tilde{\mathbf{x}}$ by changing the bits at a single vertex (or a small set of vertices) so as to eliminate at least one violating edge while otherwise choosing the new bits to (approximately) minimize the lexicographic term. The decrease in the adjacency contribution from $\mathbf{x}^*$ to $\tilde{\mathbf{x}}$ is at least $A_\textnormal{adjacency}$ (since the number of violating edges is reduced by at least one), while the increase in the lexicographic energy is at most the maximal possible range of the lexicographic term:

\begin{equation}
H_\textnormal{lexicographic}(\tilde{\mathbf{x}})-H_\textnormal{lexicographic}(\mathbf{x}^*) \le |V|\cdot\sum_{k=1}^L P_k
\end{equation}

\noindent because each $s_k$ can change by at most $|V|$ in absolute value.

By Inequality \eqref{equation:adjacency} we have:

\begin{equation}
A_\textnormal{adjacency} > |V|\cdot\sum_{k=1}^L P_k
\end{equation}

\noindent so the net change in total energy satisfies:

\begin{align}
H^\textnormal{logarithmic}_\textnormal{MGC}(\tilde{\mathbf{x}})&-H^\textnormal{logarithmic}_\textnormal{MGC}(\mathbf{x}^*) \le\nonumber \\ & -A_\textnormal{adjacency} + |V|\cdot \sum_{k=1}^L P_k < 0
\end{align}

\noindent Thus $\tilde{\mathbf{x}}$ has strictly smaller total energy than $\mathbf{x}^*$, so $\mathbf{x}^*$ could not have been a lowest-energy solution. This contradiction proves every lowest-energy solution must satisfy $H_\textnormal{adjacency}=0$, \textit{i.e.} be feasible, proving (ii).

\end{proof}

\begin{corollary}[Practical Logarithmic HUBO Penalties for \textsc{Minimum Graph Coloring}]\label{corollary:logarithmic-practical}

A convenient explicit choice that satisfies the conditions of Theorem \ref{theorem:penalties_coloring} is:

\begin{align}
P_k &= (|V|+1)^{k-1},\;\forall\;k=1,\dots,L\\
A_\textnormal{adjacency} &= |V|\cdot\sum_{k=1}^L P_k + 1
\end{align}

\noindent For this, $\sum_{j=1}^{k}P_j \le (|V|+1)^{k}/n$. Hence $P_{k+1}=(|V|+1)^{k} > |V|\cdot\sum_{j=1}^{k}P_j$, and the bound on $A_\textnormal{adjacency}$ is immediate.
\end{corollary}

\subsection{Costs of Reducing HUBO Models to QUBO Models}

The HUBO of Equation \eqref{equation:min_graph_coloring_hubo} has degree up to $2L$ where $L=\lceil\log_2 C_\textnormal{num}\rceil$. To execute on hardware supporting only quadratic interactions, we reduce to QUBO form by the Rosenberg method \cite{rosenberg1972breves}, which replaces a product of $d$ binary variables with $d-2$ auxiliary variables enforced by quadratic penalties (Equations~\eqref{equation:rosenberg_1}--\eqref{equation:rosenberg_2}).

\subsubsection{The Rosenberg Reduction Method}

The Rosenberg method reduces higher-order polynomial terms to quadratic form by introducing auxiliary binary variables. For a product of $|V|$ binary variables $\prod_{i=1}^{n} x_i$, the method requires $n-2$ auxiliary variables and proceeds recursively. Specifically, to compute the product $x_1 x_2 \cdots x_n$, we introduce auxiliary binary variables $a_1, a_2, \ldots, a_{n-2}$ and enforce the constraints:

\begin{equation}\label{equation:rosenberg_1}
    a_1 = x_1 x_2,\; a_2= a_1 x_3,\; \cdots,\; a_{n-2} = a_{n-3} x_n
\end{equation}

\noindent where the final auxiliary variable $a_{n-2}$ represents the desired product. Each constraint $a_i = a_{i-1} x_{i+1}$ can be encouraged using a quadratic penalty term of the form:

\begin{align}\label{equation:rosenberg_2}
    M_i\cdot (a_i &- a_{i-1} x_{i+1})^2 = \nonumber \\ & M_i \cdot (a_i - 2a_i a_{i-1} x_{i+1} + a_{i-1} x_{i+1})
\end{align}

\noindent where we have used binary variable idempotency.

\subsubsection{Application to Minimum Graph Coloring}

For each edge $(u,v)\in E$, the adjacency term $\prod_{k=1}^L(x_{u,k}\odot x_{v,k})$ is reduced in two stages:

\smallskip

\noindent\textbf{Stage 1.} For each edge and bit position $k$, introduce $y_{u,v,k}=x_{u,k}\odot x_{v,k}$, enforced by a penalty $M^{(1)}_{u,v,k}\cdot(y_{u,v,k}-1+x_{u,k}+x_{v,k}-2x_{u,k}x_{v,k})^2$. This adds $|E|\cdot L$ auxiliary variables.

\smallskip

\noindent\textbf{Stage 2.} The resulting product $\prod_{k=1}^L y_{u,v,k}$ is reduced by Rosenberg substitution with $L-2$ additional auxiliaries per edge, enforced by penalties $\{M^{(2)}_{u,v,i}\}$.

The total auxiliary variable count is $N_{\text{aux}}=|E|\cdot(2L-2)$, giving a total logical qubit count of $|V|\cdot L + |E|\cdot(2L-2)$. This is less than the one-hot count $(|V|+1)\cdot C_\textnormal{num}$ when:

\begin{equation}\label{equation:quadratization_condition}
    |E| < \frac{1}{2}\cdot\frac{(|V|+1)\cdot C_\textnormal{num} - \lceil\log_2 C_\textnormal{num}\rceil}{\lceil\log_2 C_\textnormal{num}\rceil - 1}
\end{equation}

\subsubsection{Setting the Quadratization Penalty Coefficients}\label{section:quadratization_penalty}

The penalty coefficients $\{M^{(1)}_{u,v,k}\}$ for $k\in[1,2,\dots,L]$ and $\{M^{(2)}_{u,v,i}\}$ for $i\in[1,2,\dots,L-2]$ must be chosen large enough to ensure properties \textnormal{(i)--(ii)} of Theorem~\ref{theorem:penalties_coloring}.

\begin{theorem}[Sufficient Rosenberg Penalties for Logarithmic \textsc{Minimum Graph Coloring} HUBO to QUBO Reduction]\label{theorem:rosenberg_penalties}
    Suppose the penalties  $\{M^{(1)}_{u,v,k}\}$ for $k\in[1,2,\dots,L]$ and $\{M^{(2)}_{u,v,i}\}$ for $i\in[1,2,\dots,L-2]$ satisfy the following two inequalities:

    \begin{align}
        M^{(1)}_{u,v,k} &> A_\textnormal{adjacency} + |V| \cdot\sum_{j=1}^{L} P_j \\
        M^{(2)}_{u,v,i} &> A_\textnormal{adjacency} + |V| \cdot \sum_{j=1}^{L} P_j
    \end{align}

    \noindent Then any lowest-energy solution $\mathbf{x}^*$ of $H^\textnormal{logarithmic}_\textnormal{MGC}(\mathbf{x})$ satisfies properties \textnormal{(i)--(ii)} of Theorem~\ref{theorem:penalties_coloring}.
\end{theorem}

\begin{proof}
    Each auxiliary variable constraint contributes a penalty that is zero when satisfied and at least $M$ when violated. The maximum benefit of any single constraint violation is bounded by the full range of the remaining Hamiltonian terms: at most $A_\textnormal{adjacency}$ from the adjacency term and $|V|\cdot\sum_{j=1}^L P_j$ from the lexicographic term. Since each $M$ exceeds this bound, no violation can reduce total energy.
\end{proof}

\begin{corollary}[Practical Rosenberg Penalties]\label{corollary:rosenberg-practical}
    A convenient uniform choice using the values of Corollary~\ref{corollary:logarithmic-practical} is:
    \begin{equation}
        M^{(1)}_{u,v,k} = M^{(2)}_{u,v,i} = 2\left((|V|+1)^L - 1\right) + 1 + \epsilon
    \end{equation}
    for any $\epsilon > 0$.
\end{corollary}

\noindent The bounds presented here are sufficient but not necessary. Alternative reduction techniques such as compressed quadratization \cite{Mandal2020} may further reduce auxiliary variable overhead.

\subsection{$2$-Qubit Gate Requirements}\label{section:gate_counts}

We derive the CNOT gate count for one QAOA phase-separation layer $e^{-i\gamma H}$ under each encoding. A $k$-local Pauli-$Z$ phase gadget requires $2(k-1)$ CNOTs \cite{cowtan2019phase}. We set $L=\lceil\log_2 C_\textnormal{num}\rceil$ and use the Ising substitution $x_j=(I-Z_j)/2$.

\begin{theorem}[CNOT Gate Count for One-Hot QUBO Encoding]\label{theorem:cnot_onehot}
    The one-hot QUBO formulation requires:
    \begin{equation}
        N_{\mathrm{CNOT}}^{\mathrm{(one-hot)}} = C_\textnormal{num}\cdot\left(|V|(C_\textnormal{num}+1)+2|E|\right)
    \end{equation}
\noindent CNOT gates.
\end{theorem}

\begin{proof}
    Since the formulation is quadratic, each monomial $x_ix_j$ yields one $2$-local Ising interaction requiring $2$ CNOTs. Counting distinct quadratic monomials: the one-hot constraint contributes $|V|\cdot\binom{C_\textnormal{num}}{2}$, the adjacency constraint contributes $|E|\cdot C_\textnormal{num}$, and the link constraint contributes $|V|\cdot C_\textnormal{num}$ (the count term is purely linear). The total is $\frac{1}{2}|V|C_\textnormal{num}(C_\textnormal{num}-1)+C_\textnormal{num}(|E|+|V|)$; multiplying by $2$ gives the result.
\end{proof}

\begin{theorem}[CNOT Gate Count for Logarithmic HUBO Encoding]\label{theorem:cnot_hubo}
    The logarithmic HUBO formulation requires:
    \begin{equation}
        N_{\mathrm{CNOT}}^{\mathrm{(log)}} = |E|\cdot\left(2(L-1)\cdot 2^L + 2\right)
    \end{equation}
\noindent CNOT gates.
\end{theorem}

\begin{proof}
    The lexicographic term is $1$-local and contributes zero CNOTs. For each edge, the Ising substitution yields $x_{u,k}\odot x_{v,k}=(1+Z_{u,k}Z_{v,k})/2$, so the adjacency product becomes $2^{-L}\cdot \prod_{k=1}^L(1+Z_{u,k}Z_{v,k})$. Expanding gives $\sum_{S\subseteq[L]}\prod_{k\in S}Z_{u,k}Z_{v,k}$, where a subset of size $s\geq 1$ produces a $2s$-local interaction requiring $2(2s-1)$ CNOTs. Summing over all nonempty subsets per edge:
    \begin{align}
        \sum_{s=1}^{L}\binom{L}{s}\cdot 2(2s-1) &= 4\cdot L\cdot 2^{L-1} - 2(2^L-1) \nonumber\\
        &= 2(L-1)\cdot 2^L + 2
    \end{align}
    using $\sum_{s=0}^L\binom{L}{s}s=L\cdot 2^{L-1}$ and $\sum_{s=0}^L\binom{L}{s}=2^L$. Multiplying by $|E|$ gives the result.
\end{proof}

\noindent\textbf{Comparison.} The one-hot encoding requires $\Theta(C_\textnormal{num}^2\cdot|V|+C_\textnormal{num}\cdot|E|)$ CNOTs versus $\Theta( C_\textnormal{num}\cdot\log_2 C_\textnormal{num}\cdot|E|)$ for the logarithmic encoding. For sparse graphs ($|E|=\Theta(|V|)$), the ratio scales as $\Theta(C_\textnormal{num}/\log_2 C_\textnormal{num})$, reflecting a near-exponential gate reduction. Even for dense graphs, the one-hot encoding retains an additional $C_\textnormal{num}^2\cdot|V|$ contribution from the one-hot constraint, ensuring asymptotically fewer $2$-qubit gates for the logarithmic encoding whenever $|V|\geq\Omega(\lceil\log_2 C_\textnormal{num}\rceil)$.

\section{Extension of Logarithmic HUBO Encoding to General Graph Partitioning Problems}\label{section:treating_the_general_case}

We now cast the general graph partitioning (GP) objective function of Equation \eqref{equation:general_graph_partitioning_problem} in terms of binary variables by assigning $L=\lceil\log_2|V|\rceil$ bits $\{x_{v,1},x_{v,2},\dots,x_{v,L}\}=\{x_{v,k}\}_{k\in[1,...,L]}$ to each vertex $v\in V$. Our first step is to recognize that:

\begin{align}
    \mathbf{1}[\ell_u= \ell_v]&=\prod_{k=1}^Lx_{u,k}\odot x_{u,k} \\
     \mathbf{1}[\ell_u\neq \ell_v]&=1-\prod_{k=1}^Lx_{u,k}\odot x_{u,k}
\end{align}

\noindent Thus we represent $\phi_{u,v}(\mathbf{1}[\ell_u= \ell_v])$ in terms of binary variables. Next, instead of directly \textit{counting} the number of distinct labels assigned to the vertices of $G$, as in $|\{\ell_v:v\in V\}|$, we regularize by way of a lexicographic penalty:

\begin{equation}
    H_\textnormal{lexicographic}(\mathbf{x})=\sum_{k=1}^{L}P_k\sum_{v\in V}x_{v,k}
\end{equation}

\noindent where $P_1<P_2<\cdots<P_L$ (see Section \ref{section:sufficient_conditions_lex_and_feas} for strict bounds on $\{P_k\}_{k\in[1,2,\dots,L]}$). The full HUBO representation of Equation \eqref{equation:general_graph_partitioning_problem} is then given by:

\begin{equation}\label{equation:general_hubo}
    H^\textnormal{logarithmic}_\textnormal{GP}(\mathbf{x})=H_\textnormal{partition}(\mathbf{x})+H_\textnormal{lexicographic}(\mathbf{x})
\end{equation}

\noindent where:

\begin{align}
    H_\textnormal{partition}(\mathbf{x})=&A_\textnormal{partition}\cdot\sum_{(u,v)\in E}\Bigg[\alpha_{u,v}\cdot\prod_{k=1}^L(x_{u,k}\odot x_{v,k})\nonumber\\& +\beta_{u,v}\cdot\left(1-\prod_{k=1}^L(x_{u,k}\oplus x_{v,k})\right)\Bigg]
\end{align}

\noindent $A_\textnormal{partition}$ is a penalty associated with the partition term, and $P_k$ are the lexicographic penalties.

\subsection{Sufficient Conditions for Ensuring Lexicographic Minimization and Feasibility of the Lowest-Energy Solution}\label{section:sufficient_conditions_on_feasiblity}

We have already established in Section \ref{section:sufficient_conditions_lex_and_feas} sufficient conditions for guaranteeing that lexicographic minimization is preserved (part (i) of Theorem \ref{theorem:penalties_coloring}); the proof applies to all graph partitioning problems of the form presented in Equation \eqref{equation:general_hubo}. Regarding $A_\textnormal{partition}$ for such problems, feasibility depends on the structure of the particular problem. For example, in the case of \textsc{Minimum Graph Coloring} and \textsc{Minimum} $k$-\textsc{Cut}, feasible solutions correspond to valid colorings and valid $k$-cuts, respectively. In the case of \textsc{Community Detection}, all solutions are feasible given the absence of hard constraints.

In what follows, we therefore approach determining sufficient conditions on $A_\textnormal{partition}$ for graph partitioning problems with hard constraints, assuming only that we can easily separate the solution space into feasible and infeasible solutions. For graph partitioning problems lacking hard constraints, it is sufficient to set $A_\textnormal{partition}>0$.

\begin{mydef}[]{Feasibility Gap for Graph Partitioning Problems with Hard Constraints}{feasibility_gap}
    Let $\mathcal{F} \subseteq \{0,1\}^{n \times L}$ denote the set of feasible binary solutions (those satisfying all hard constraints of the problem). Define:
    
    \begin{align}
        H_{\min}^{\text{feasible}} &= \min_{\mathbf{x} \in \mathcal{F}} H_{\text{partition}}(\mathbf{x}) \\
        H_{\min}^{\text{infeasible}} &= \min_{\mathbf{x} \notin \mathcal{F}} H_{\text{partition}}(\mathbf{x})
    \end{align}
    
    $\Delta_{\min} = H_{\min}^{\text{infeasible}} - H_{\min}^{\text{feasible}}$ is the \textit{feasibility gap}.
\end{mydef}

\begin{theorem}[Sufficient Partition Penalty for Logarithmic Graph Partitioning HUBOs with Hard Constraints]
    Suppose the penalties $P_1, \ldots, P_L$ and $A$ satisfy:
    
    \begin{align}
        P_{k+1} &> |V| \cdot \sum_{j=1}^{k} P_j \quad \text{for every } k = 1, \ldots, L-1 \\
        A_\textnormal{partition} &> \frac{|V| \cdot \sum_{k=1}^{L} P_k}{\Delta_{\min}}\label{equation:general_penalty}
    \end{align}
    
    \noindent Then any lowest-energy solution $\mathbf{x}^*$ of $H^\textnormal{logarithmic}_\textnormal{GP}(\mathbf{x})$ (Equation \eqref{equation:general_hubo}) is feasible; \textit{i.e.}, $\mathbf{x}^* \in \mathcal{F}$ (Definition \ref{def:feasibility_gap}).
\end{theorem}

\begin{proof}

Suppose for contradiction that $\mathbf{x}^*$ is a lowest-energy solution with $\mathbf{x}^* \notin \mathcal{F}$. Since $\mathbf{x}^*$ is infeasible, we have $H_{\text{partition}}(\mathbf{x}^*) \geq H_{\min}^{\text{infeasible}}$. Consider any feasible solution $\tilde{\mathbf{x}} \in \mathcal{F}$. By definition, $H_{\text{partition}}(\tilde{\mathbf{x}}) \geq H_{\min}^{\text{feasible}}$. 

The change in total energy from $\mathbf{x}^*$ to $\tilde{\mathbf{x}}$ is:
\begin{align}
    H^\textnormal{logarithmic}_\textnormal{GP}(\tilde{\mathbf{x}}) &- H^\textnormal{logarithmic}_\textnormal{GP}(\mathbf{x}^*) \nonumber \\ &= A_\textnormal{partition} \cdot [H_{\text{partition}}(\tilde{\mathbf{x}}) - H_{\text{partition}}(\mathbf{x}^*)] 
    \nonumber \\&\quad\,+ [H_{\text{lex}}(\tilde{\mathbf{x}}) - H_{\text{lex}}(\mathbf{x}^*)] \notag \\
    &\leq A_\textnormal{partition} \cdot [H_{\min}^{\text{feasible}} - H_{\min}^{\text{infeasible}}] 
    + |V| \cdot \sum_{k=1}^{L} P_k \notag \\
    &= -A_\textnormal{partition} \cdot \Delta_{\min} + |V| \cdot \sum_{k=1}^{L} P_k
\end{align}

\noindent where the inequality uses $H_{\text{partition}}(\tilde{\mathbf{x}}) \geq H_{\min}^{\text{feasible}}$ and $H_{\text{partition}}(\mathbf{x}^*) \geq H_{\min}^{\text{infeasible}}$. Recall from earlier that the lexicographic term can change by at most $|V|n \cdot \sum_{k=1}^{L} P_k$ (as each $s_k$ can change by at most $|V|$).

By Inequality \eqref{equation:general_penalty}, we have:
\begin{equation}
A_\textnormal{partition} \cdot \Delta_{\min} > |V| \cdot \sum_{k=1}^{L} P_k
\end{equation}

\noindent Therefore:
\begin{equation}
H^\textnormal{logarithmic}_\textnormal{GP}(\tilde{\mathbf{x}}) - H^\textnormal{logarithmic}_\textnormal{GP}(\mathbf{x}^*) < 0
\end{equation}

\noindent This implies $H^\textnormal{logarithmic}_\textnormal{GP}(\tilde{\mathbf{x}}) < H^\textnormal{logarithmic}_\textnormal{GP}(\mathbf{x}^*)$, contradicting the assumption that $\mathbf{x}^*$ is a lowest-energy solution. We conclude that every lowest-energy solution must be feasible.
\end{proof}

This requires us to use $\Delta_{\min}$ to calculate an appropriate value for $A_\textnormal{partition}$, in the general case. Fortunately, as we have seen, in at least one case (\textsc{Minimum Graph Coloring}) we can quickly compute this value.

\section{Benchmarking on Quantum Hardware}\label{section:benchmarking}

We generated 350 randomly-connected graphs with edge densities in $[0.2, 0.8]$ and $|V|\in[4,100]$ (available at \url{https://github.com/pangara/random-connected-graphs/}). Each graph was treated as a \textsc{Minimum Graph Coloring} instance and solved on a quantum annealer using both the one-hot formulation (Section \ref{section:one_hot_formulation}, penalties per Section \ref{section:onehot_penalties}) and the quadratized HUBO formulation (Section \ref{section:hubo_formulation}, penalties per Sections~\ref{section:sufficient_conditions_lex_and_feas} and~\ref{section:quadratization_penalty}). Brooks' bound on chromatic number was pre-computed in linear time to upper-bound qubit requirements per encoding. The code used to conduct benchmarking, analysis, and plotting is available at \url{https://github.com/nooloop/QuantumGraphColoring}.

\subsection{Quantum Annealing}

The \textsc{D-Wave} \texttt{Advantage2\_system1.11} quantum processing unit (QPU) contains \textasciitilde{}4,400 superconducting flux qubits on a Zephyr graph with $20$-way connectivity \cite{boothby2021zephyr}. Problems whose interaction graphs are not subgraphs of this topology required \emph{minor embedding}, where each logical qubit is represented as a chain of ferromagnetically-coupled physical qubits \cite{pinilla2024context, cai}. We used \textsc{D-Wave}'s default \texttt{minorminer} algorithm for embedding; embedding time is excluded from Time-to-Solution measurements but the embeddings were recorded for qubit count analysis. The QPU implements the transverse-field Ising Hamiltonian $H(t)=A(t)H_0+B(t)H_1$, where $A(t)$ decreases from $1$ to $0$, $B(t)$ increases from $0$ to $1$, $H_0=-\sum_i\sigma_x^{(i)}$ is the initial driver Hamiltonian, and $H_1=-\sum_i h_i\sigma_z^{(i)}-\sum_{i<j}J_{ij}\sigma_z^{(i)}\sigma_z^{(j)}$ is the problem Hamiltonian.

\subsubsection{QPU Settings}\label{section:settings_of_the_dwave_quantum_processing_unit_used_for_benchmarking}

All problem instances were solved entirely on-chip without hybrid decomposition. The annealing time was $T_S=20\;\mu$s (comparable to the quoted $5$--$20\;\mu$s coherence time), using the default schedule with chain strength set by uniform torque compensation \cite{Gilbert2024}. For each instance and encoding, $1000$ runs were executed with thermalization time ${\approx}\,1000\;\mu$s/run and readout time ${\approx}\,(40+N_q)\;\mu$s/run, where $N_q$ is the physical qubit count. HUBO-to-QUBO reduction was performed by \texttt{dimod.make\_quadratic()} with penalties set per Section \ref{section:quadratization_penalty}.

\subsubsection{Timing}\label{section:timing_the_execution_of_the_dwave_quantum_processing_unit}

Per-problem execution time is the sum of programming, annealing, readout, and thermalization times across all runs. Embedding time is treated as negligible (given that embeddings are cacheable \textit{via} hash-table lookup). Exact per-problem timing was retrieved from the \texttt{SampleSet.info['timing']} attribute.

\begin{figure}
\begin{center}
    \includegraphics[width=8.5cm]{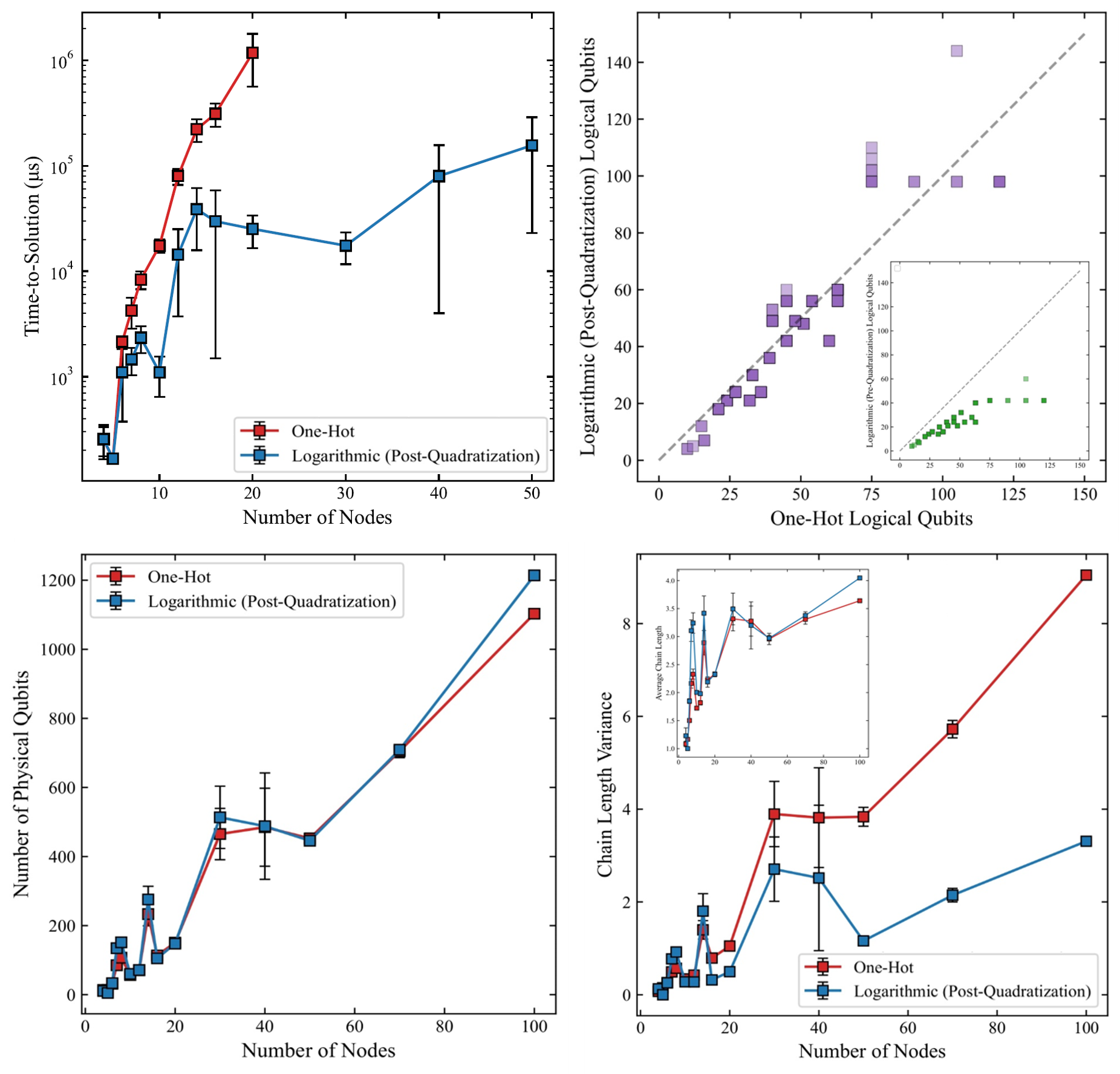}
\caption{\textit{Upper left}: TTS as a function of $|V|$. \textit{Upper right}: Post-quadratization logical qubits per encoding; inset: pre-quadratization. Intensity corresponds to the number of problems occupying the same point. \textit{Lower left}: Physical qubits after minor embedding. \textit{Lower right}: Chain length variance; inset: average chain length.}
\label{fig:dwave}
\end{center}
\end{figure}

\subsubsection{Performance Measures}\label{section:performance_measures}

We evaluate the encodings using three complementary metrics: Time-to-Solution, which captures end-to-end solver performance; survival-time analysis, which enables unbiased aggregation across instances where one or both encodings fail to find the optimum; and qubit counts at each stage of the compilation pipeline.

\paragraph{Time-to-Solution (TTS)} TTS is the expected wall-clock time to find the lowest-energy feasible solution (identified across both encodings) with probability $\geq 50\%$:

\begin{equation}\label{equation:tts}
    \mathrm{TTS}
    = T_\textnormal{programming}+T_\textnormal{run}\cdot \frac{\ln(1 - 0.5)}{\ln(1 - p_S)}
\end{equation}

\noindent where $p_S$ is the fraction of $R=1000$ runs returning the target solution, and $T_\textnormal{run}=T_{\mathrm{anneal}} + T_{\mathrm{readout}}+T_\text{thermalize}$. Instances with $p_S=0$ are right-censored and treated using survival-time analysis when aggregating TTS results over problems.

\paragraph{Survival-Time Analysis for Aggregated TTS} When aggregating TTS across instances grouped by a shared characteristic (\textit{e.g.}, same $|V|$ or edge density), undefined TTS values ($p_S=0$) cannot be discarded without downward bias or assigned arbitrary large values \cite{Gagliolo2010}. We apply Kaplan--Meier survival analysis \cite{kaplan1958nonparametric}: computed TTS values enter as uncensored observations and unsolved instances enter as right-censored at $T_{\mathrm{censor}}=R\cdot(T_{\mathrm{anneal}}+T_{\mathrm{readout}}+T_\text{thermalize})$. The median TTS is extracted at $\hat{S}(t)=0.5$ with $95\%$ confidence intervals \textit{via} Greenwood's formula \cite{Miettinen2008}; when the curve does not cross $0.5$, the median is reported as a lower bound.

\paragraph{Qubit Counts} For each encoding and instance we record: (i)~pre-quadratization logical qubits, (ii)~post-quadratization logical qubits (HUBO only), and (iii)~physical qubits after minor embedding.

\section{Benchmarking Results}\label{section:results}

Figure~\ref{fig:dwave} compares the one-hot and logarithmic encodings on the \textsc{D-Wave} \texttt{Advantage2\_system1.11} QPU across TTS, qubit counts, and chain length statistics.

\paragraph{Time-to-Solution} Both encodings exhibit increasing TTS with $|V|$ (upper-left panel). The logarithmic encoding achieves substantially lower TTS throughout. While comparable at $|V|=4$, TTS is approximately one to two orders of magnitude lower by $|V|=20$. This widening gap suggests that the logarithmic encoding grants progressively greater advantages as problem size grows, likely due to more uniform minor-embedding chains (see below).

\paragraph{Logical and Physical Qubit Counts} The upper-right inset confirms that pre-quadratization logical qubits lie well below the $y{=}x$ diagonal for every instance, verifying the expected exponential compression. After quadratization (main scatter), many instances remain below the diagonal, but dense graphs can exceed the one-hot count; this is consistent with Equation \eqref{equation:quadratization_condition}. The lower-left panel shows that \textit{physical} qubit counts \textit{after} minor embedding are comparable between encodings; the denser logarithmic encoding interaction graph absorbs much of its pre-quadratization advantage after quadratization and embedding.

\paragraph{Chain Length Statistics}\label{section:chain_lengths} While average chain lengths are similar between encodings (lower-right inset), chain length \textit{variance} is markedly lower for the logarithmic encoding, with the gap widening with $|V|$. We posit this is a key driver of the TTS advantage: non-uniform chains distort the problem Hamiltonian by imposing different effective energy scales on different logical qubits \cite{pelofske2025comparing}; a single global chain strength cannot simultaneously serve both long and short chains \cite{Gilbert2024}; and longer outlier chains are disproportionately susceptible to thermal noise, seeding chain breaks \cite{jeong2025embedding, pelofske2025comparing}. The logarithmic encoding's more uniform chains allow the chain strength to be closer to optimal for all chains simultaneously, preserving problem fidelity throughout the anneal. These results suggest that chain length \textit{uniformity} can be decisive even when physical qubit counts and average chain lengths are comparable.

\section{Discussion and Conclusions}\label{section:discussion_and_conclusions}

We have presented a one-hot QUBO encoding, and a qubit- and gate-efficient HUBO formulation for label-symmetric, pairwise-decomposable graph partitioning problems with minimization of partition count. To the best of our knowledge, this work is the first to approach the optimization of these problems using quantum computing. The HUBO construction encodes each $L$-valued variable using $\lceil\log_2 L\rceil$ bits and employs a novel lexicographic penalty system that implicitly minimizes the number of distinct labels without dedicated indicator variables. We derived provably sufficient conditions on all penalty coefficients---including quadratization penalties---guaranteeing feasibility and optimality of the lowest-energy solution. The formulation applies broadly, encompassing problems satisfying the properties listed in Section \ref{sec:graph_paritioning_problems}.

\paragraph{Theoretical Contributions} The lexicographic penalty system we introduce is novel. It allows our logarithmic encoding of graph partitioning problems to reduce per-vertex bits from $L$ (one-hot) to $\lceil\log_2 L\rceil$, yielding a near-exponential CNOT gate reduction for sparse graphs and asymptotically fewer gates whenever $|V|\geq\Omega(\lceil\log_2 L\rceil)$. To the best of our knowledge, we provide the first sufficient penalty conditions for a logarithmically encoded graph partitioning HUBO, including quadratization penalties. The Rosenberg reduction introduces $|E|\cdot(2L{-}2)$ auxiliaries; quadratization requires fewer logical qubits than one-hot when the graph is sufficiently sparse (Equation \eqref{equation:quadratization_condition}). Gate-based optimizers can execute the HUBO without quadratization.

\paragraph{Experimental Results} The logarithmic encoding achieves TTS one to two orders of magnitude lower than one-hot on \textsc{D-Wave} \texttt{Advantage2\_system1.11} for graphs with $20$ nodes, with advantage demonstrating growth with problem size. This improvement persists even when post-quadratization logical qubits exceed the one-hot count, indicating that performance gains are not solely attributable to qubit savings. Our qubit analysis confirms that dense instances can exceed the one-hot count after Rosenberg reduction (consistent with Equation \eqref{equation:quadratization_condition}), while sparser instances retain substantial savings. Although physical qubit counts are comparable between encodings, considerably lower chain-length variance under logarithmic encoding likely drives its superior performance.

\paragraph{Generality} The formulation applies to any problem in the class defined in Section \ref{section:treating_the_general_case}; the only problem-specific quantity is the feasibility gap $\Delta_{\min}$. It is compatible with quantum annealing, QAOA, VQE, QITE, and GAS.

\paragraph{Limitations and Future Work} The sufficient penalty conditions may be conservative; tighter bounds could reduce Hamiltonian coefficient dynamic range and improve analog hardware performance \cite{Verma2022}. Alternative quadratization strategies could extend the qubit advantage to denser graphs. Empirical extension to larger instances and other problem classes (\textit{e.g.}, minimum $k$-cut, balanced partitioning) would further establish the approach's practical scope. Benchmarking using gate-based hardware and algorithms remains to be conducted. Finally, our lexicographic penalization system may be adapted to domain-wall encoding of variables; sufficient penalty analysis and benchmarking on quantum hardware remain.

\section*{Conflict of Interest Statement}

\scriptsize{VKM is a co-founder and shareholder of Menten AI, a peptide design company that uses advanced optimization strategies.}

\newpage

\bibliographystyle{IEEEtran}
\bibliography{IEEEabrv, main}

\end{document}